\documentclass[a4paper]{article}

\usepackage{fullpage}
\usepackage{authblk}
\usepackage{amsthm}

\usepackage{graphicx}
\usepackage{lipsum}
\usepackage{paralist}

\usepackage{amssymb,amsfonts,amsmath}
\usepackage[utf8]{inputenc}
\usepackage[T1]{fontenc}
\usepackage{microtype}
\usepackage{cleveref}
\usepackage{enumerate}
\usepackage{xspace}
\usepackage{nicefrac,xcolor}
\usepackage{comment}

\newcommand{\upse}{UPSE\xspace}
\newcommand{\cons}{consistent\xspace}
\newcommand{\centralPath}{backbone\xspace}

\newtheorem{theorem}{Theorem}
\newtheorem{lemma}{Lemma}

\newtheorem{observation}{Observation}

\crefname{figure}{Figure}{Figures}
\crefname{lemma}{Lemma}{Lemmas}
\crefname{theorem}{Theorem}{Theorems}
\crefname{observation}{Observation}{Observations}

\usepackage[caption=false]{subfig}

\begin{document}
\title{Upward Point Set Embeddings of Paths and Trees\thanks{This work was initiated at the 7th Annual Workshop on Geometry and Graphs, March 10--15, 2019, at the Bellairs Research Institute of McGill University, Barbados. We are grateful to the organizers and to the participants for a wonderful workshop.}}

\author[1]{Elena Arseneva\thanks{Email: \texttt{e.arseneva@spbu.ru}. Partially supported by by RFBR, project \mbox{20-01-00488}.}}
\author[2]{Pilar Cano\thanks{Email: \texttt{pilar.cano@ulb.ac.be}. Supported by the F.R.S.-FNRS under Grant no.~MISU F 6001 1.}}
\author[3]{Linda Kleist\thanks{Email: \texttt{kleist@ibr.cs.tu-bs.de}.}}
\author[4]{Tamara Mchedlidze\thanks{Email: \texttt{mched@iti.uka.de}.}}
\author[5]{Saeed Mehrabi\thanks{Email: \texttt{smehrabi@mun.ca}.}}
\author[6]{Irene Parada\thanks{Email: \texttt{i.m.de.parada.munoz@tue.nl}.}}
\author[7]{Pavel Valtr\thanks{Email:\texttt{valtr@kam.mff.cuni.cz}. Work supported by grant no.~19-17314J of the Czech Science Foundation (GA\v CR).}}

\affil[1]{Saint Petersburg State University, Russia}
\affil[2]{Universit\'e Libre de Bruxelles, Belgium}
\affil[3]{Technische Universit\"at Braunschweig, Germany}
\affil[4]{Karlsruhe Institute of Technology, Germany}
\affil[5]{Memorial University, St. John's, Canada}
\affil[6]{TU Eindhoven, The Netherlands}
\affil[7]{Department of Applied Mathematics, Charles University, Czech Republic}

\maketitle              
\begin{abstract}
We study upward planar straight-line embeddings (\upse) of directed trees on given point sets. The given point set $S$ has size at least the number of vertices in the tree. 
For the special case where the tree is a path $P$ we show that: (a) If $S$ is one-sided convex, the number of \upse{s} equals the number of maximal monotone paths in  $P$. (b) If $S$ is in general position and $P$ is composed by three maximal monotone paths, where the middle path is longer than the other two, then it always admits an \upse on $S$. We show that the decision problem of whether there exists an \upse of a directed tree with~$n$ vertices on a fixed point set $S$ of $n$ points is NP-complete, by relaxing the requirements of the previously known result which relied on the presence of cycles in the graph, but instead fixing position of a single vertex.
Finally, by allowing extra points, we guarantee that each directed caterpillar on $n$ vertices and with $k$ switches in its \centralPath admits an \upse on every set of $n 2^{k-2}$ points.

\end{abstract}

\section{Introduction}
A classic result by Fary, Stein, Wagner \cite{fary1948,stein1951,wagner1936}, known as Fary's theorem, states that every planar graph has a crossing-free straight-line drawing. Given a directed graph (called \emph{digraph} for short), it is natural to represent the direction of the edges by an \emph{upward} drawing, i.e., every directed edge is represented by a monotonically increasing curve. Clearly, it is necessary for the digraph to be acyclic in order to allow for an upward drawing. 

In nice analogy to Fary’s theorem,  if a planar digraph has an upward planar drawing, then it also allows for an upward planar straight-line drawing~\cite{di1988}.
In contrast, not every planar acyclic digraph has an upward planar drawing~\cite{di1988}. Nevertheless, some classes of digraphs always allow for such drawings. For instance, every directed tree has an upward planar straight-line drawing~\cite[p.~212]{tollis1999}.

In this work, we study upward planar straight-line drawings of directed trees on given point sets. 
An \emph{upward planar straight-line embedding} (\upse, for short) of a digraph $G=(V,E)$ on a point set~$S$, where $|V|\leq |S|$, is an injection from $V$ to $S$
such that the induced straight-line drawing is planar (crossing-free) and upward, i.e., for the directed edge $uv \in E(G)$, the point representing $u$ lies below the point representing $v$.
As point-set embeddings of planar undirected  graphs~\cite{bose2002PSembedding,JGAA-2,JGAA-132,EGC17_UPSE,pach1991PSembedding}, \upse{s} have been an active subject of research~\cite{AngeliniFGKMS10,BinucciGDEFKL10,KaufmannMS13,JupwardPath14}. In the following we review the state of the art in relation to our problems.

Kaufmann et al.~\cite{KaufmannMS13} showed that in case $|V|=|S|$ it is NP-complete to decide whether an upward planar digraph admits an \upse on a given point set. 
We note that the upward planar digraph obtained in their reduction contains cycles in its underlying undirected graph.  
The same authors 
gave polynomial-time algorithm 
to decide 
if an upward planar digraph  admits an \upse on a convex point set.

A digraph whose underlying undirected structure is a simple path is called \emph{oriented path}.
For the class of oriented paths multiple partial results have been provided, by either limiting the class of the point set, the class of oriented paths, or by considering the case where $|S|>|V|$. In particular, by limiting the class of point sets, Binucci et al. showed that every oriented path admits an \upse on every convex point set of the same size~\cite{BinucciGDEFKL10}. 
By limiting the class of oriented paths, it has been shown that the following subclasses of oriented paths always admit an \upse on any general point set of the same size: 
\begin{compactenum}
\item An oriented path with at most five switches\footnote{A vertex of a digraph which is either a source or a sink is called \emph{switch}.} and at least two of its sections\footnote{A section of an oriented path is a subpath that connects two consecutive switches.} having length\footnote{The \emph{length} of a path is the
number of vertices in it.} two~\cite{AngeliniFGKMS10}.
\item An oriented path $P$ with three switches~\cite{BinucciGDEFKL10}.
\item An oriented path $P = (v_1,\dots, v_n)$, so that if its vertex $v_i$ is a sink,
then its vertex $v_{i+1}$ is a source~\cite{BinucciGDEFKL10}.
\item An oriented path $P$ such that its decomposition into maximal monotone paths $P_1,P_2 \ldots , P_r$ satisfies that
$|P_i| \ge \sum_{j>i} |P_j| $ for every $i= 1,2,\ldots, r-1$~\cite{EGC17_UPSE}. 
\end{compactenum}

Given these partial results, it is an intriguing open problem whether every oriented path $P$ 
admits an \upse on every general point set $S$ with $|P|=|S|$. 
In contrast to this question, there exists a directed tree $T$ with $n$ vertices and a set $S$ with $n$ points in convex position such that $T$ does not admit an \upse on $S$~\cite{BinucciGDEFKL10}. While restricting the class of trees to directed caterpillars, Angelini et al.~\cite{AngeliniFGKMS10} have shown that an \upse  exists on any convex point set.

The variant of the problem where the point set is larger than the oriented path, was considered by Angelini et al.~\cite{AngeliniFGKMS10} and Mchedlidze~\cite{JupwardPath14}. They proved that every oriented path $P$ with $n$ vertices  and $k$
switches admits an \upse on every general point set with $n 2^{k-2}$
points~\cite{AngeliniFGKMS10} or on  $(n-1)^2+1$ points Mchedlidze~\cite{JupwardPath14}, respectively. 

\paragraph{Our Contribution.} 
In this paper we continue the study of \upse of digraphs. We tackle the aforementioned open problem from multiple sides. Firstly, we show that the problem of deciding whether a digraph on $n$ vertices admits an \upse on a given set of $n$ points remains NP-complete even for trees when one vertex lies on a predefined point. (Section~\ref{sec:hardness}). 
This strengthens the previously known NP-completeness result, where the underlying undirected structure contained cycles~\cite{KaufmannMS13}.  Thus, even if it is still possible that every oriented path admits an \upse on every general point set, this new NP-completeness might foreshadow that a proof for this fact will not lead to a polynomial time construction algorithm.
Secondly, we provide a new family of $n$-vertex oriented paths that admit an \upse on any general set of $n$ points  (Section~\ref{sec:specialPaths}), extending the previous partial results~\cite{AngeliniFGKMS10,BinucciGDEFKL10}. 
Thirdly, by aiming to understand the degrees of freedom that one has while embedding an oriented path on a point set, we show that the number of different \upse{s} of an $n$-vertex oriented path on a one-sided convex  set of $n$ points is equal to the number of sections the path contains (Section~\ref{sec:realizations}).
Finally, as a side result, we study the case where the point set is larger than the graph and show that the upper bound $n 2^{k-2}$ on the size of a general point set that hosts every oriented path~\cite{AngeliniFGKMS10} can be extended to caterpillars (Section~\ref{sec:caterpillars}), where $k$ is the number of switches in the caterpillar. The proof is largely inspired by the corresponding proof for oriented paths. However, the result itself opens a new line of investigations of providing upper bounds on the size of general point sets that are sufficient for \upse of families of directed trees.

\paragraph{Definitions.} 
A set of points is called
\emph{general}, if no three of its points lie on the same
line and no two have the same $y$-coordinate. 
The \emph{convex hull} $H(S)$ of a point set $S$ is the point set that can be obtained as a convex combination of the points of $S$. We say that a point set is in \emph{convex position}, or is a \emph{convex point set}, if none of its points lie in the convex hull of the others. 
Given a general point set $S$, we denote the lowest and the highest point of $S$ by $b(S)$ and
$t(S)$, respectively.  
A subset of points of a convex point set $S$ is called \emph{consecutive} if its points appear consecutively as we traverse the convex hull of $S$.
A convex point set $S$ is called \emph{one-sided} if the points $b(S)$ and $t(S)$ are consecutive in it; refer to~\cref{fig:convChain}.

Let   $\Gamma$ be an \upse of a digraph $G=(V,E)$ on a point set $S$.  For every $v\in V$, $\Gamma(v)$ denotes the point of $S$ where vertex $v$ has been mapped to by $\Gamma$. 
A \emph{directed tree} is a digraph, whose underlying graph is a tree. 
A digraph, whose underlying graph is a simple path is called \emph{oriented path}.
A \emph{directed caterpillar} is a directed tree in which the removal of the vertices of degree 1  results in an oriented path: the \emph{\centralPath}. 
For an oriented path $(v_1,v_2,\dots,v_n)$,  we call $v_i v_{i+1}$ a \emph{forward} (resp., \emph{backward}) edge if it 
is oriented from $v_i$ to $v_{i+1}$ (resp., from $v_{i+1}$ to $v_i$).
A vertex of a digraph with
in-degree (resp., out-degree) equal to zero is called a \emph{source} (resp., \emph{sink}). A vertex of a digraph which is either a source or a sink is called \emph{switch}. A subpath of an oriented path $P$ connecting two of its consecutive switches is said to be \emph{monotone} and called a \emph{section} of $P$. A section is \emph{forward} (resp., backward) if it consists of forward (resp., backward) edges.

\section{Counting Embeddings of Paths on Convex Sets}
\label{sec:realizations}
In this section, we study the number of \upse{s} that an $n$-vertex oriented path has on a one-sided convex set of $n$ points. We show that this number is equal to the number of sections in the oriented path. We start with the following

\begin{lemma}\label{obs:one-vertex-per-realization}
Let $P$ be an $n$-vertex oriented path with $v_1$ being one of its end-vertices and let $S$ be a one-sided convex point set with $|S|=n$. For any two different \upse{s} $\Gamma_1$ and $\Gamma_2$ of $P$ on $S$, it holds that $\Gamma_1(v_1) \neq \Gamma_2(v_1)$.
\end{lemma}
\begin{proof}
Let $\{p_1,\dots,p_n\}$ be the points of $S$ sorted by $y$-coordinate. For the sake of contradiction, assume that there exist two different \upse{s} $\Gamma_1$ and $\Gamma_2$ of $P$ on $S$ with $\Gamma_1(v_1) = \Gamma_2(v_1)$. Additionally, assume that the considered counterexample is minimal, in the sense that for the vertex $v_2$ of $P$, adjacent to $v_1$, $\Gamma_1(v_2) \neq \Gamma_2(v_2)$. 
By~\cite[Lemma~3]{BinucciGDEFKL10}, vertices $v_1$ and $v_2$ lie on consecutive points of $S$. 
We assume that the edge $v_1v_2$ is a forward edge; 
the case when $v_1v_2$ is a backward edge 
is symmetric.
Conditions $\Gamma_1(v_1) = \Gamma_2(v_1)$,  $\Gamma_1(v_2) \neq \Gamma_2(v_2)$, and the fact that $v_1$, $v_2$ lie on consecutive points of $S$, imply that $\Gamma_1(v_1)=\Gamma_2(v_1)=p_1$, and $\Gamma_i(v_2)=p_n$, $\Gamma_j(v_2)=p_2$, with $i,j\in \{1,2\}$. Embedding $\Gamma_i$ is an \upse of $P$ when the edge $v_2v_3$ is backward, while $\Gamma_j$ is an \upse of $P$ when the edge $v_2v_3$ is forward. We arrive to a contradiction. 
\end{proof}

We are now ready to prove the result of this section.

\begin{theorem}\label{lem:convChain}
An $n$-vertex oriented path $P$ with $k$ sections has exactly  $k$ \upse{s} on a one-sided convex set of $n$ points.
\end{theorem}
\begin{proof}
We first show that $P$ has at least $k$ \upse{s}. To do so, let $P=(v_1,\dots,v_n)$ be an oriented path with $k$ sections and let $S$ be a one-sided convex point set with points $\{p_1,\dots,p_n\}$ ordered by the increasing $y$-coordinate.
Let $v_l$ be the switch of $P$ preceding~$v_n$. Thus, the subpath of $P$ between $v_l$ and $v_n$ is $P$'s last section. Let denote the subpath of $P$ between $v_1$ and $v_l$ by $P'$. We prove the statement of the theorem by the following stronger induction hypothesis: there are $k$ \upse{s} of $P$ on $S$, such that one of them maps $v_n$ to $p_n$ (resp., $p_1$) if the last section is forward (resp., backward). 
The base case of $k=1$ is trivial. Assume that the $k$-th section of $P$ is forward (resp., backward) and let $S'=\{p_1,\dots,p_l\}$ (resp., $S'=\{p_{n-l+1}, \ldots, p_n\}$). 
By induction hypothesis, the path $P'$ has $k-1$ \upse{s} on $S'$, with one of them mapping $v_l$ to $p_1$ (resp., $p_n$).
Let $\Gamma$ be one of them. Assume that $\Gamma(v_l)=p_i$, $p_i \in S'$. 
We shift every vertex that has been mapped to point $p_j$, $i<j\leq l$ (resp., $n-l+1 \leq j<i$) by $n-l$ points up (resp., down); refer to~\cref{fig:convChain}(a). We map the $k$-th section to points $p_{i+1},\dots,p_{i+n-l}$ (resp., $p_i-1,\dots,p_{i-n+l}$). This gives us $k-1$ \upse{s} of $P$ on~$S$.

\begin{figure}[tb]
	\centering
	\includegraphics[page=1]{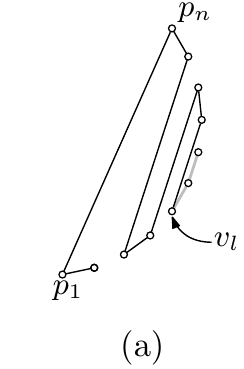}\hfil
	\includegraphics[page=2]{ConvexChain.pdf}\hfil
	\includegraphics[page=7]{nShape-2}
	\caption{(a--b) An illustration for the proof of \cref{lem:convChain}, showing an extension of an \upse of $P'$ (black) to an \upse of $P$ (gray). In (a) $v_l$ lies on a non-extreme point of $S'$, while in (b) $v_l$ is an extreme point of~$S'$. (c) Description of path $P$ presented in \cref{lem:nShapeSameLength}.}
	\label{fig:convChain}
\end{figure}

Recall that, by induction hypothesis, $P'$ has an \upse $\Gamma'$ on $S'$ that maps $v_{l}$ to $p_1$ (resp., $p_n$) on $S'$, since its last section is backward (resp., forward). Thus, we can also extend $\Gamma'$ by mapping the $k$-th section to points $p_{l+1}, \ldots, p_n$ (resp., $p_1, \dots, p_{n-l}$); refer to~\cref{fig:convChain}(b). Hence, there exists a \upse of $P$ that maps $v_n$ to $p_n$(resp., $p_1$) if the last section is forward (resp., backward). 
By  Lemma~\ref{obs:one-vertex-per-realization}, no two of the constructed embedding of $P$ on $S$ are the same. Thus, $P$ has at least~$k$ \upse{s} on $S$. 

We now apply a counting argument to show that each oriented path with $k$ sections has exactly $k$ \upse{s} on $S$. Note that the total number of possible \upse{s} of different directed paths of size $n$ on an $n$-point one-sided convex point set is $n\cdot 2^{n-2}$. To see this, note that an \upse can be encoded by the start point and the position (clockwise or counterclockwise consecutive) of the next point. 
For the last choice the clockwise and the  counterclockwise choices coincide. Moreover, the number of oriented paths with $n-1$ edges and $k$ sections is $\rho_k:=2 \binom{n-2}{k-1}$ and the number of $n$-vertex oriented paths is $\sum_{k=1}^{n-1}\rho_k=2^{n-1}$. Let $\eta_k$ denote the number of \upse{s} of all oriented paths on $n$ vertices 
with $k$ sections.

As shown above, an oriented path with $k$ sections has at least $k$ \upse{s} on $S$. By the symmetry of the binomial coefficient, it holds that $\rho_{k}=\rho_{n-k}$. Therefore,  the number of \upse{s} of all oriented paths with $k$ and $n-k$ sections evaluates to at least $k\rho_k+(n-k)\rho_{n-k}=\frac{n}{2} (\rho_k+\rho_{n-k})$.
This implies
\begin{align*}
    n 2^{n-2}
    &=\sum_{k=1}^{n-1}\eta_k 
    \geq \sum_{k=1}^{n-1} k \rho_k
    = \frac{n}{2} \cdot \sum_{k=1}^{n-1}\rho_k 
    = \frac{n}{2} \cdot 2^{n-1}
    =n2^{n-2}.
\end{align*}
Consequently, each oriented path with $k$ sections has  exactly $k$ \upse{s} on $S$. 
\end{proof}


\section{Embedding of Special Directed Paths}
\label{sec:specialPaths}
In this section we present a family of oriented paths that always admit an \upse on every general point set of the same size. For an illustration of paths in this family, consider \cref{fig:convChain}(c).

\begin{theorem}
\label{lem:nShapeSameLength}
Let $P$ be an 
oriented path with three sections $P_1=(u_1,\dots,u_a)$, $P_2=(v_1,\dots,v_b)$ and $P_3=(w_1,\dots,w_c)$, where $u_a=v_1$, $v_b=w_1$ and $b\geq a,c$. Then path $P$ admits an \upse on any general set $S$ of $n=a+b+c-2$ points.
\end{theorem}

\begin{proof}
We assume that $P_1$ is forward, otherwise, we rename the vertices of $P$ by reading it off from vertex $w_c$.
Let $\ell$ denote the line through $t(S)$ and $b(S)$, and let $\ell^-$ (resp., $\ell^+$) be the halfplane on the left (resp., on the right) when walking from $b(S)$ to $t(S)$ along $\ell$. We assume that both $\ell^-$  and $\ell^+$ 
are closed, and hence $|S\cap \ell^+|+|S\cap \ell^-|=a+b+c$. We now consider two cases. 

\begin{figure*}[tb]
\centering
\includegraphics[page=4]{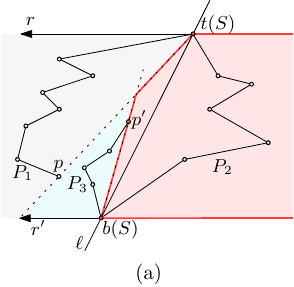}\hfil
\includegraphics[page=3]{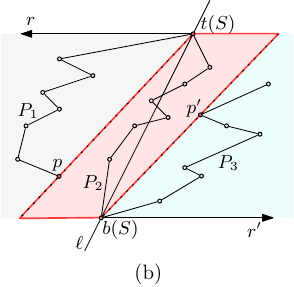}
\caption{An illustration of the proof of Theorem~\ref{lem:nShapeSameLength}. (a) Embedding of $P$ in Case 1. (b) Embedding of $P$ in Case 2. Point set $S_2$ lies in the red convex polygon.  }%
\label{fig:nShapeNEW}
\end{figure*} 

\noindent{\bf Case 1:} $|S\cap \ell^+|\geq a+c$ or $|S\cap \ell^-|\geq a+c$. 
We consider the case $|S\cap \ell^-|\geq a+c$ (refer to~\cref{fig:nShapeNEW}(a), with $|S\cap \ell^+|\geq a+c$ being symmetric. We rotate a left heading horizontal ray~$r$ emanating from $t(S)$ in counter-clockwise direction until it hits the $a$-th point (including $t(S)$); we let $p$ denote this point and let $S_1$ denote the points swept by $r$.  Notice that $b(S)\not \in S_1$ because $|S\cap \ell^-|> a$. We embed $P_1$ on the points in~$S_1$, by sorting them by ascending $y$-coordinate. 

From $|S\cap \ell^-|\geq a+c$ it follows that $|(S\setminus S_1)\cap \ell^-|\geq c$. Hence, we can  embed~$P_3$ on $S\cap \ell^-$. To do so, we rotate around $b(S)$ a left-heading horizontal ray $r'$ in clockwise direction until it hits the $c$-th point of $(S\setminus S_1)\cap \ell^-$; denote these $c$ points by $S_3$ and let $p'$ be the last point hit by $r'$. 
We embed $P_3$ on the points of $S_3$ by sorting them by ascending $y$-coordinate. 

Let $S_2=(S\setminus (S_1\cup S_3))\cup\{b(S),t(S)\}$. Then, the polygonal region determined by the horizontal lines through $t(S)$ and $b(S)$, the line through $t(S)$ and $p$, the line through $b(S)$ and $p'$, and the vertical line through the rightmost point of $S_2$ is a convex region that contains all points in $S_2$ and 
and has $t(s)$ (resp., $b(s)$) as its topmost (resp., bottommost) vertex.
 Therefore, we can embed $P_2$ onto $S_2$ by sorting them by descending $y$-coordinate. We observe that $u_a=v_1$ and $v_b=w_1$ have been consistently embedded on $b(S)$ and $t(S)$, respectively.

\noindent{\bf Case 2:} $|S\cap \ell^+|\leq a+c$ and $|S\cap \ell^-|\leq a+c$.  It follows that  $|S\cap \ell^-|\geq b\geq a$ and $|S\cap \ell^+|\geq b\geq c$; refer to~\cref{fig:nShapeNEW}(b).
Since $|S\cap \ell^-|\geq a$, we can construct the set $S_1\subset S\cap \ell^-$ similarly to Case~1, and embed $P_1$ on its points. We then rotate a right-headed horizontal ray $r'$ in counter-clockwise direction around $b(S)$ until it hits the $c$-th point $p'$ of $\ell^+$ and denote by $S_3$ the points swapped by $r'$. Since $|S\cap \ell^+|> c$, $t(S) \not \in S_3$. 
We embed $P_3$ onto $S_3$ by sorting the points by ascending $y$-coordinate.

Finally, let $S_2=(S\setminus (S_1\cup S_3))\cup\{b(S),t(S)\}$. We note that the polygonal region determined by the horizontal lines through $t(S)$ and $b(S)$, the line through $t(S)$ and $p$, and the line through $b(S)$ and $p'$ is a convex polygon that contains all points of $S_2$. Also recall that  $b(S)$ and $t(S)$ are respectively the bottommost and the topmost  points of $S_2$. Thus, we can embed $P_2$ onto the points in $S_2$ by  sorting them by descending $y$-coordinate. Note that $u_a=v_1$ and  $v_b=w_1$ have been consistently mapped to $t(S)$ and $b(S)$, respectively.
\end{proof}

\section{Embedding Caterpillars on Larger Point Sets}
\label{sec:caterpillars}
In this section, we provide an upper bound on the number of points that suffice to construct an \upse of an $n$-vertex caterpillar. We first introduce some necessary notation. Let $C$ be a directed caterpillar with $n$ vertices, $r$ of which, $v_1,v_2,\dots,v_r$, form its \centralPath. For each vertex $v_i$ ($i=1,2,\dots,r$), we denote by $A(v_i)$ (resp., $B(v_i)$) the set of degree-one vertices of $C$ adjacent to $v_i$ by outgoing (resp., incoming) edges. Moreover, we let $a_i=|A(v_i)|$ and $b_i=|B(v_i)|$.
\begin{theorem}
\label{thm:caterpillars}
Let $C$ be a directed caterpillar with $n$ vertices and $k$ switches in its \centralPath. Then $C$ admits an \upse on any general point set $S$ with $|S|\geq n2^{k-2}$.
\end{theorem}
\begin{proof}
Assume that $C$ contains $r$ vertices $v_1,v_2,\dots,v_r$ in its \centralPath. Let $C_{\alpha,\beta}$ denote the subgraph of $C$ induced by the vertices $\bigcup_{i=\alpha}^{\beta}{\{v_i\}\cup A(v_i) \cup B(v_i)}$. We first make the following observation.
\begin{observation}
\label{obs:caterpillar}
If the \centralPath of $C$ contains exactly one section, $C$ has an \upse on any general set $S$ of $n$ points.
\end{observation}

{
To see that, we sort the points in $S=\{p_1,\dots, p_n\}$ by ascending $y$-coordinate 
and,  
assuming $C$ is forward (the backward case is symmetric),  
we embed $v_i$ on $p_x$, where $x=i+b_i+\sum_{j=1}^{i-1}(b_j+a_j)$. See to~\cref{fig:caterpillar}(a) for an illustration.
}


\begin{figure}[t!]
	\centering
	\includegraphics[page=2]{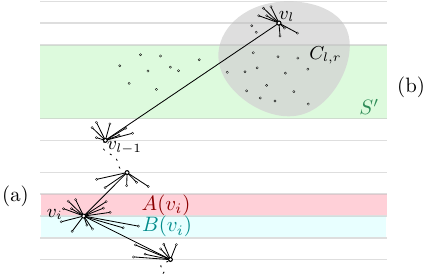}
	\caption{Illustration of the proofs of (a) Observation~\ref{obs:caterpillar}; (b) Theorem~\ref{thm:caterpillars}.}
	\label{fig:caterpillar}
\end{figure}

We now prove the theorem by induction on $k$.
We assume as the {induction hypothesis} that a directed caterpillar $C$ with $i$ switches and $n_i$ vertices has an \upse on any general point set $S$ of at least $n_i2^{i-2}$ points. 
Let $v$ denote the first vertex of the \centralPath of $C$. 
We additionally assume that if $v$ is a source (resp. sink), then $v$ is mapped on the $(|B(v)|+1)$-th bottommost (resp. $(|A(v)|+1)$-th topmost)  point of $S$.
For $k=2$, the \centralPath of $C$ contains one section; hence, the induction hypothesis follows from Observation~\ref{obs:caterpillar} and its proof. 

We now consider a caterpillar $C$ with $i+1$ switches and $n_{i+1}$ vertices. Let $v_1$ and $v_l$ denote the first and the second switches of the \centralPath of $C$, respectively. Let $S$ be a set of at least $N=n_{i+1}2^{i-1}$ points. In the following, we only consider the case where the \centralPath of $C_{1,l}$ is forward; 
the backward case is symmetric.

We construct an \upse of $C_{1,l-1}$ on the $c_1 = \sum_{i=1}^{l-1}{(1+b_i+a_i)}$ lowest points of $S$ by applying Observation~\ref{obs:caterpillar}. Let $p$ denote the point where $v_{l-1}$ is mapped. Let~$S'$ denote the unused points of $S$; thus, $|S'|=N-c_1$.  
We have $n_{i+1}=n_i+c_1$, where $n_i$ is the number of vertices of $C_{l,r}$. Recall that $N=n_{i+1}2^{i-1}$. Therefore, $|S'|>n_i 2^{i-1}$. 
Let $p'$ be the $(a_l+1)$-th topmost point in $S'$ and let $\ell$ denote the line through $p$ and $p'$. Line $\ell$ partitions $S'$ into two sets, so that for the largest, say $S''$, it holds that $|S''|\geq n_i2^{i-2}$.
 Let $S'''$ be the union of $S''$ with the set of points lying above
$v_l$.
 Since $C_{l,r}$ contains $i$ switches, by induction hypothesis, we can construct an \upse of $C_{l,r}$ on $S'''$ such that $v_l$ is mapped on the $(a_l+1)$-th topmost point of $S'''$, which is the point $p'$. The only edge of the drawing of $C_{1,l}$ that interferes with the drawing of $C_{l,r}$ is $(v_{l-1},v_l)$; however, the drawing of $C_{l,r}$ (except for the edges incident to~$v_l$) lies on one side of the edge $(v_{l-1},v_l)$. 
 See to~\cref{fig:caterpillar}(b) for an illustration.
\end{proof}

\section{NP-Completeness for Trees}
\label{sec:hardness}
In this section, we consider the following problem:
Given a directed tree $T$ with $n$ vertices, a vertex $v$ of~$T$, a set~$S$ of $n$ points in the plane, and a point $p$ in~$S$, does $T$ have an \upse on~$S$ 
which  maps $v$ to $p$?
Our goal is to show the following. 
\begin{theorem}
	\label{thm:np-1}
	\upse of a directed tree with one fixed vertex is $NP$-complete. 
\end{theorem}

The 3-Partition problem is a strongly NP-complete problem by Garey and Johnson~\cite{garey1975complexity,garey1979computers}, which is formulated as follows: 
For a given multiset of $3m$ integers $\{a_1,\ldots, a_{3m}\}$, does there exist a  partition into $m$ triples $(a_{11}, a_{21}, a_{31}),\ldots,$ $(a_{1m},a_{2m},a_{3m})$, so that for each $j \in[m]$, $\sum_{i=1}^3a_{ij} = b$, where 
$b = (\sum_{i=1}^{3m}a_{i})/m$. 
In this section, we present a reduction from 3-Partition. 
Without loss of generality, we may assume, possibly multiplying each $a_i$ by 3, 
that each $a_i$ is divisible by 3.
\begin{figure}[htb]
	\centering
	\includegraphics[page=4]{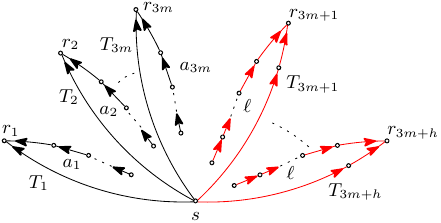}
	\caption{An illustration of the tree $T$ in the proof of Theorem~\ref{thm:np-1}.}
	\label{fig:np-tree}
\end{figure}

Given $3m$ integers $\{a_1,\ldots, a_{3m}\}$, where each $a_i$ is divisible by 3, we construct an instance of our problem as follows.  
Let $\ell$ and 
$h$ be two large numbers such that $\ell>mb$ and $h=m+1$. We construct a tree~$T$ with $n:=mb+h\ell + 1$ vertices, 
and a point set~$S$ with $n$ points.  
As illustrated in \cref{fig:np-tree}, (the undirected version of) the tree $T$ is a subdivision of a star, that is, $T$ has a single vertex~$s$ of degree greater than two. Specifically, the degree of vertex  $s$ is $3m+h$, i.e., $3m+h$ paths meet in their initial vertex $s$; we call each such path a \emph{branch} of~$T$.
Let $T_1,\ldots, T_{3m+h}$ denote the branches, respectively.  
For $i\in\{1,\dots, 3m\}$, the branch~$T_i$ is a path with $a_i$ edges and called \emph{small}; for $i\in\{3m+1,\dots,3m+h\}$, the branch $T_i$ is a path with $\ell$ edges and called \emph{large}. Note that $a_i$ and $\ell$ may also be interpreted as the number of vertices of a branch that are different from $s$.

For each branch $T_i$, we define its \emph{root} $r_i$ as follows. For a small branch $T_i$, $r_i$ is the first vertex on $T_i$ that is different from $s$; for a large branch, $r_i$ is the second vertex different from $s$. 
For every branch, all edges are oriented so that its root is the unique sink of the branch.
Thus the sinks of $T$ are exactly  the $3m+h$ roots of the branches, and the  sources of $T$ are $s$ and the $3m+h$ leaves of $T$.


The point set $S$ is depicted in \cref{fig:np-pointset} and constructed as follows.
The lowest point of $S$ is $p = (0,0)$.
Let $E$ be an ellipse with center at $p$, (horizontal) semi-major axis $5$ and (vertical) semi-minor axis $3$. 
Let $C$ be a convex $x$- and $y$- monotone curve above $E$ and also above the point $(-3,3)$, represented by the filled square mark in \cref{fig:np-pointset}.
Consider the cone defined by the upward rays from~$p$ with slopes -1, -2, and subdivide this cone into $2m+1$ equal sectors by
$2m$ upward rays from $p$. 
\begin{figure}[b!]
	\centering
	\includegraphics[page = 11]{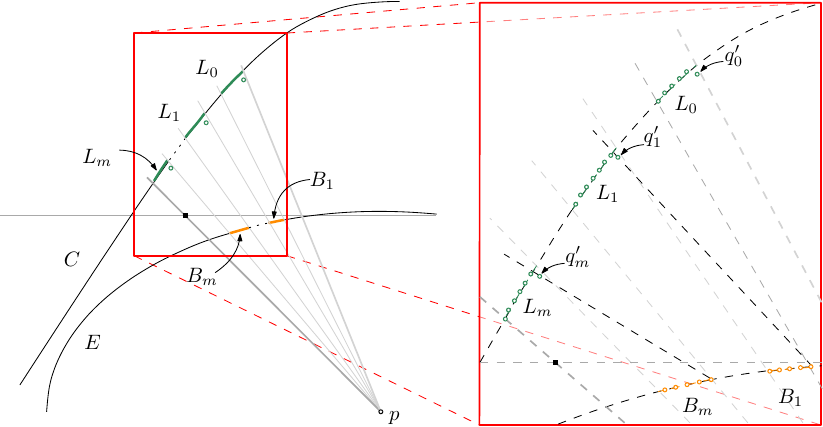}
	\caption{An illustration of the point set $S$ in the proof of Theorem~\ref{thm:np-1} and a zoomed portion of $S$.}
	\label{fig:np-pointset}
\end{figure}
Let $s_0,\ldots,s_{2m}$ denote the 
obtained sectors ordered counter-clockwise.  
For each odd $i = 2k+1$ with  $0 < i \leq 2m$,
consider the intersection between the sector $s_i$ and the ellipse $E$ (the orange arcs in~\cref{fig:np-pointset}) and let $B_k$ be the set of $b$ equally spaced points on this intersection. 

For each even $i = 2k$ with  $0 \leq k \leq m$,
consider the intersection $c_k$ between sector $s_i$ and curve $C$ (the green arcs in~\cref{fig:np-pointset}). Let the point set $L_k$ be constructed as follows. 
Place the first point $q_k'$ of $L_k$ in the sector $s_i$ slightly to the right of (and thus slightly below, to stay inside the sector) the topmost point of the arc~$c_k$. 
Let, for $k>0$, $x_k$ be the point of intersection between $c_k$ and the line through~$q_k'$ and the topmost point of the set $B_k$. For $k = 0$ let $x_k$ be the highest point of~$c_k$. 
Place the second point  $q_k$ of $L_k$ on $c_k$ slightly below $x_k$, but still above $q_k'$, see also \cref{fig:np-pointset}. Place $\ell-2$ points equally spaced on $c_k$ below~$q_k$.
This concludes our description of $T$ and $S$. 
In the remaining, the point sets $B_k$ are called \emph{small sets}, and the point
sets $L_k$ are called \emph{large sets}.

To prove~\cref{thm:np-1},  we show that $\{a_1,\ldots a_{3m}\}$ admits a 3-partition if and only if $T$ admits an \upse $\Gamma$ on $S$ where $\Gamma(s)=p$. In particular, our proof is based on the fact that if such an \upse exists, then it nearly fulfills the special property of being a \emph{\cons embedding}  of $T$ on $S$ (\cref{lemma:nice}). 
An embedding $\Gamma$ is \emph{\cons} if  $\Gamma(s)=p$  and each large branch of $T$ is mapped to exactly one large point set; for a schematic illustration consider~
\cref{fig:nice-drawing}.
\begin{figure}[tb]
	\centering
	\includegraphics[page=2]{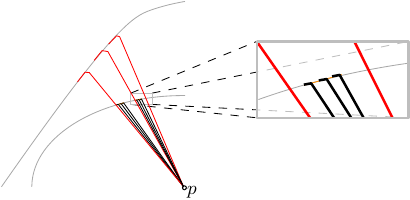}
	\caption{A schematic illustration of a \cons embedding of $T$ on $S$.}
	\label{fig:nice-drawing}
\end{figure}

By construction of $S$, if the large branches are mapped to the large point sets, then the small branches of $T$ must be subdivided into $m$ groups each with a total of $b$ vertices.  Each such group then corresponds to a triple $(a_{1i},a_{2i},a_{3i})$ of the partition of $(a_1,\ldots,a_{3m})$ that sums up to $b$. Conversely, a 3-partition directly yields a \cons embedding. This proves the following fact:
\begin{observation}
	$T$ has a \cons embedding on $S$ if and only if $\{a_1,\ldots a_{3m}\}$ has a 3-partition.
\end{observation}

Now we prove the main lemma of this section in order to conclude~\cref{thm:np-1}.

\begin{lemma} \label{lemma:nice}
	If $T$ admits an \upse $\Gamma$ 
	on $S$ mapping $s$ to $p$ ($\Gamma(s) = p$), then 
	$T$ also admits a \cons  \upse on $S$ mapping $s$ to $p$.
\end{lemma}
\begin{proof}
	To show that 
	$T$ admits a \cons embedding, 
	we work our way from right to left. Let $\Gamma_{i}$ denote the partial embedding of $T$ induced by restricting~$\Gamma$ to the points in $S_i:=\cup_{k\leq i}(L_k\cup B_k)$  for every $i$. We say $\Gamma_{i}$  is \cons if large branches are mapped to large sets of~$S_i$ and small branches are mapped to small sets of~$S_i$.
	In particular, we prove the following:
	If there is an embedding $\Gamma$ with $\Gamma(s) = p$ whose partial embedding $\Gamma_i$ is \cons, then there is an embedding~$\Gamma'$ with $\Gamma'(s) = p$ whose partial embedding $\Gamma'_{i+1}$ is \cons.

	For unifying notation, we define $B_0,B_{-1},L_{-1}:=\emptyset$. Suppose there is an embedding $\Gamma$ such that $\Gamma_i$ is \cons.
	Suppose that the partial embedding $\Gamma_{i+1}$ of $\Gamma$ is not \cons (otherwise let $\Gamma'$ be $\Gamma$). 
	Since $q_{i+1}$ is the highest point of $S\setminus S_i$, $\Gamma(r_j)=q_{i+1}$ for some branch $T_j$.  
	We distinguish two cases depending on whether $T_j$ is a small or large branch.
	If $T_j$ is a small branch, then $B_{i+1}$ and~$q_{i+1}'$ is to the right of segment $pq_{i+1}$. Depending on whether $T_j$ continues left or right, a set of small branches is mapped on $B_{i+1}\cup\{q_{i+1}'\}$ or $B_{i+1}\cup\{q_{i+1},q_{i+1}'\}$, respectively. (In the latter case, $T_j$ is one of the small branches mapped on $B_{i+1}\cup\{q_{i+1},q_{i+1}'\}$.)  However, the cardinality of these sets is not divisible by three; a contradiction. Therefore $T_j$ is a large branch. 
	
	Then there exists a point $y \neq q_{i+1}'$ such that the segments $py$ and $yq_{i+1}$ represent the first two edges of~$T_j$. If $y$ belongs to a large set, then $yq_{i+1}$ separates all points between  $y$ and $q_{i+1}$ on $C$ from~$p$. 
	Therefore, there are only three different options for the placement of $y$ on a large set: $q_{i+2}$, $q'_{i+2}$ or the left neighbour of $q_{i+1}$ in $L_{i+1}$. 
	If $y$ is $q_{i+2}$ or $q'_{i+2}$, then a number of small branches are mapped to $B_{i+1}\cup\{q'_{i+2}, q'_{i+1}\}$ or to $B_{i+1}\cup\{q'_{i+1}\}$, respectively. However, the cardinality of these sets is not divisible by three, a contradiction. 
	In the case that $y$ is the left neighbour of $q_{i+1}$ in $L_{i+1}$, branch $T_j$ must continue to the right of $q_{i+1}$. However, this would imply that $T_j$ contains at most $b+3$ points, a contradiction to $\ell>b+3$. 
	Therefore, $y$ belongs to a small set $B_j$ with $j> i+1$. Moreover, $y$ belongs to $B_{i+2}$, otherwise $q'_{i+2}$ had to be the root of a large branch which would lead to a contradiction. 
	Let $A$ be the set of all points of $B_{i+2}$ to the right of $y$. 
	The set $A \cup \{q_{i+1}'\}$ is separated from the rest of the points of $S\setminus (S_i \cup \Gamma(T_j) \cup B_{i+1})$ by the segments $py$ and $yq_{i+1}$.  
	Since there are less than $\ell$ points in $A\cup \{q'_{i+1}\}\cup B_{i+1}$,  several small branches of $T$ are mapped to $A \cup \{q'_{i+1}\}\cup B_{i+1}$. In addition, since $|B_{i+1}\cup\{q_{i+1}'\}|$ is not divisible by three, for the 
	small branch that maps its root to $q'_{i+1}$, the rest of this branch is mapped to the left of $pq'_{i+1}$. 
	To construct a sought embedding $\Gamma'$ we modify $\Gamma$ as follows. Consider the point~$y'$ in 
	$B_{i+2}$ ($y'$ lies to the right of $y$)
	such that the segment $y'q'_{i+1}$ is in the partial embedding of $\Gamma_{i+1}$.  
	First, we map to $q_{i+1}'$ the vertex of $T_j$ previously mapped to $y$ (i.e., $\Gamma^{-1}(y)$). Let $A' \subset A$ be all the points to the left of $y'$ in $A$. For each vertex mapped to a point $z$ in $A'\cup\{y'\}$, we map such vertex to the left neighbour of $z$ in $A'\cup\{y\}$. Finally, the vertex that was mapped to $q'_{i+1}$ is now mapped to $y'$. See Fig.~\ref{fig:reviewer-case}.
	
	\begin{figure}[t!]
		\centering
		\includegraphics[page = 14]{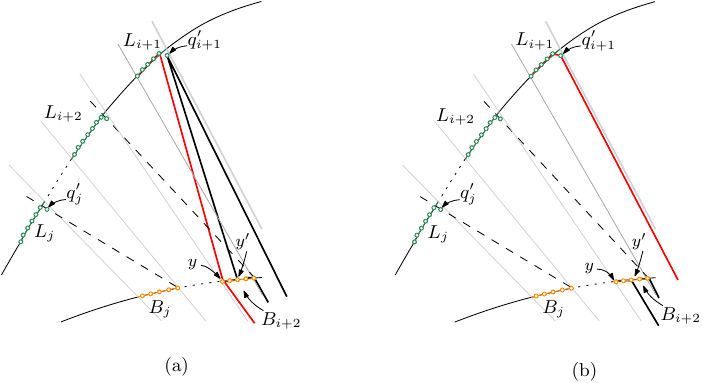}
		\caption{(a) A partial embedding 
			$\Gamma_{i+1}$ of $\Gamma$ that is not \cons, while $\Gamma_i$ is. (b) A modification $\Gamma'$ of $\Gamma$ such that  $\Gamma'_{i+1}$ is \cons.}
		\label{fig:reviewer-case}
	\end{figure}
	
	The obtained embedding $\Gamma'$ is such that $\Gamma'_{i+1}$ is \cons. Indeed, in $\Gamma'$ the branch $T_j$ is embedded entirely on $L_{i+1}$; otherwise some point of $L_{i+1}$ would be separated from $p$.
	And since $|B_{i+1}|=b<\ell$, a set of small branches is embedded on $B_{i+1}$.
	Together with the obvious fact that in any \upse $\Gamma$ of $T$ on~$S$, $\Gamma_{-1}$ is \cons, the above observation implies the claim of the lemma.
\end{proof}

\section{Conclusion and Open Problems}
In this paper, we continued the study of \upse of directed graphs, specifically of paths, caterpillars, and trees. On the positive side, we showed that a certain family of $n$-vertex oriented paths admits an \upse on any general $n$-point set and that any caterpillar can be embedded on a general point set if the set is large enough.  Moreover, we provided the exact number of \upse{s} for every path on a one-sided convex set. 
On the negative side,  we proved that the problem of deciding whether a directed graph on $n$ vertices admits an \upse on a given set of $n$ points remains NP-complete even for trees when one vertex lies on a pre-defined point. 
We conclude with a list of interesting open problems:
\begin{compactenum}
    \item Given any oriented path $P$ and any general point set $S$ with $|P|=|S|$, does there exist an \upse of $P$ on $S$?
    \item If the answer to the previous question turns out to be negative, what is the smallest constant $c$ such that a set of $c$ paths on a total of $n$ vertices, has an \upse on every general point set on $n$ points? This problem could also be an interesting  stepping stone towards a positive answer of the previous question. 
    \item Given a directed tree $T$ on $n$ vertices and a set $S$ of $n$ points,  is it NP-hard to decide whether $T$ has an \upse on $S$?
    \item Can the provided upper bound on the size of general point sets that host every caterpillar be improved to polynomial in the number of vertices, following the result in~\cite{JupwardPath14}? If yes, can this be extended to general directed trees? 
\end{compactenum}

\bibliographystyle{abbrv}
\bibliography{ref}

\end{document}